\documentclass[runningheads]{llncs}
\usepackage{graphicx}
\usepackage[caption=false]{subfig}
\usepackage{comment}
\usepackage{ragged2e}
\usepackage{enumerate}
\usepackage{enumitem}
\usepackage{bm}
\usepackage{amssymb,latexsym,amsmath}     
\usepackage{graphicx}
\usepackage{color}
\usepackage{float}
\usepackage{mathtools}
\usepackage{blindtext}
\usepackage[T1]{fontenc}
\usepackage{wrapfig}
\usepackage{changepage}
\usepackage{cases}
\usepackage[misc,geometry]{ifsym}
\usepackage{hyperref}
\usepackage{orcidlink}

\allowdisplaybreaks

\begin{document}
\title{Complexity results for the Pilot Assignment problem in Cell-Free Massive MIMO 
}
\titlerunning{Computational results on the Pilot Assignment problem}
%
\author{Shruthi Prusty \inst{1}\textsuperscript{(\Letter)}\orcidlink{0009-0003-4157-2427} \and
Sofiat Olaosebikan\inst{2}\orcidlink{0000-0002-8003-7887}}
\authorrunning{S. Prusty and S. Olaosebikan}
%
\institute{Department of Mathematics, Indian Institute of Science Education and Research Pune, Pune, India \\
\email{shruthi.prusty@students.iiserpune.ac.in} \and
School of Computing Science, University of Glasgow, Glasgow, UK
\email{Sofiat.Olaosebikan@glasgow.ac.uk}\\
}
\maketitle              
\begin{abstract}
Wireless communication is enabling billions of people to connect to each other and the internet, transforming every sector of the economy, and building the foundations for powerful new technologies that hold great promise to improve lives at an unprecedented rate and scale. The rapid increase in the number of devices and the associated demands for higher data rates and broader network coverage fuels the need for more robust wireless technologies. The key technology identified to address this problem is referred to as Cell-Free Massive MIMO (CF-mMIMO). CF-mMIMO is accompanied by many challenges, one of which is efficiently allocating limited resources. In this paper, we focus on a major resource allocation problem in wireless networks, namely the Pilot Assignment problem (\textsc{PA}). We show that \textsc{PA} is strongly NP-hard and that it does not admit a polynomial-time constant-factor approximation algorithm. Further, we show that \textsc{PA} cannot be approximated in polynomial time within $\mathcal{O}(K^2)$ (where $K$ is the number of users) when the system consists of at least three pilots. Finally, we present an approximation lower bound of $1.058$ (resp.~$\epsilon|K|^2$, for $\epsilon >0$) in special cases where the system consists of exactly two (resp.~three) pilots.

\end{abstract}
\section{Introduction}
\subsection{Background: Cell-free Massive MIMO}
Wireless networks is an essential technology for enabling flexible communication and connectivity between individuals (or machines) across regions. In addition, it is transforming every sector of the economy (transportation, healthcare, education, etc.), and powerful new technologies (artificial intelligence, internet of things, etc.) are being built upon it. Cellular networks is the technology that 1G to 5G relies on \cite{IBNFL19,zhang2020prospective}. Here, the coverage area is divided up into non-overlapping cells, and we have a single Access Point (AP) coordinating data transmissions amongst user devices within its cell. 
\\ \\
As the number of devices that depend on wireless communication networks continues to grow, each needing a high connection rate and better coverage with minimal interference, this technology will no longer be suitable \cite{ngo2017cell,zhang2020prospective}. Particularly, resources (e.g. channels, power, spectrum) in a wireless network are limited and must be managed efficiently. Further, the User Equipments (UEs) at the edge of the cells get poor service due to \emph{inter-cell interference}.
\\ \\
For future wireless communications (e.g., 6G), the key technology that has the potential to enhance connectivity and provide better coverage for billions of users is referred to as Cell-Free Massive Multiple-Input Multiple-Output (CF-mMIMO). As the name suggests,  CF-mMIMO allows a device/user (UE) to be served by multiple access points (APs) that are within its range without the notion of boundaries; in contrast to the current technology, which allows each UE to be served by only one AP within a defined boundary. The goal of this network is to reduce inter-cell interferences, improve the uniform distribution of spectral efficiency amongst users and enhance network reliability \cite{ngo2017cell}. CF-mMIMO is accompanied by many challenges, one of which is how to efficiently manage limited resources (spectrum, pilot signals, energy, and power) – the so-called resource allocation problem in wireless networks \cite{CHEN2022695}. 

\subsection{The Pilot Assignment problem} \label{subsec12}
Throughout this paper, we will consider a CF-mMIMO system with $M$ single-antenna APs and $K$ $(K \ll M )$ UEs where the APs are randomly distributed in a large area. Since a large number of distributed APs jointly provide uniform service to a small number of UEs without any notion of boundaries, it is often the case that \emph{AP selection} is done for each user, such that only a subset of the APs providing service above a certain threshold are considered for any energy or spectral efficiency calculations pertaining to that user \cite{liu2020graph,ZHLW2021}. An important assumption in this paper is that AP selection is always done for any CF-mMIMO system we consider \cite{ZHLW2021}. We now turn our attention towards a major problem that hinders resource allocation in wireless communication networks, and is a resource allocation problem itself: the Pilot Assignment problem (\textsc{PA}). 
\\ \\
It is essential to acquire accurate \emph{Channel State Information (CSI)} between the UEs and the APs to reap all the benefits potentially provided by the distributed user-centric CF-mMIMO architecture \cite{CHEN2022695}. Channel estimation allows APs to process data signals from the UEs. To perform channel estimation for a user, a \emph{pilot signal} needs to be assigned to it. In the system setting, it is assumed that APs and UEs do not have \emph{a priori} CSI at the beginning of a coherent interval. The CSI is estimated in what is called a \emph{pilot training phase}, which usually happens in the uplink. Channel estimation is needed only at the beginning of a coherent interval $\tau_c$. Thus, $\tau$ pilot sequences (or signals) of length $\tau$ each are assigned to the UEs prior to uplink data transmission for channel estimation. Each UE is assigned one pilot. The received pilot information is used for channel estimation. The estimated channel is used to detect the received data, thereby allowing us to calculate the Spectral Efficiency (SE) of each UE \cite{DZ2022}. However, due to the limited length of the coherence interval, the available number of orthogonal pilot sequences is normally smaller than the number of UEs $(\tau \ll K)$ and some UEs have to reuse a given pilot. Hence, the orthogonality among the pilot sequences for all UEs is typically not achieved. The pilot reuse causes an impairment known as pilot contamination, which can degrade the system performance, by lowering the achievable uplink/downlink rates and signal-to-interference-plus-noise ratio (SINR) of the system \cite{liu2020graph,ZHLW2021}.

\subsection{Existing work and contributions of this paper} \label{subsec13}
\textsc{PA} has been extensively researched from the engineering perspective, with the aim of constructing heuristic-based algorithms that return sub-optimal solutions for this distributed architecture. The most straightforward and naive pilot allocation strategy is the \emph{random} pilot assignment scheme, in which the available pilots are assigned randomly to each user \cite{ngo2017cell}. Of course, this does not address any of the problems, including two nearby users sharing the same pilot, and thus turns out to be the worst scheme. The \emph{greedy} pilot assignment scheme proposed in \cite{ngo2017cell} works by iteratively improving the downlink rate for the worst user. However, such a method can only improve the worst user's performance at any given point, and cannot guarantee an improvement in the whole system's performance. This means that this algorithm is not guaranteed to converge stably to a global maximum value. 
\\ \\
The \emph{location-based greedy} pilot assignment scheme utilizes the location information of the users for the initial assignment of pilots in the greedy scheme, instead of random assignment. This method, however, does not prove to be very effective in practice and only promotes the throughput performances of a few users \cite{liu2020graph,ZCZQY2018}. The \emph{structured} pilot assignment scheme maximizes the minimum distance among UEs that share the same pilot using a clustering algorithm, but the implementation in a real-world cell-free massive MIMO system is hindered by the difficulty of finding the centroid APs in such practical systems \cite{AAL2018,liu2020graph}. The \emph{Maximal Increment} (MI) algorithm maximizes the achievable (downlink) rate by maximizing the increment in an iterative algorithm, but it has high time complexity \cite{Yin2019PilotAA}. The \emph{Hungarian Algorithm} based pilot assignment scheme is an iterative procedure based on the Hungarian algorithm to enhance system throughput and fairness by avoiding pilot reuse among nearby users. However, it has been observed that it is not sufficiently accurate to measure pilot contamination solely based on the geographical distance between the users \cite{BDFZF2021,ZHLW2021}. \\ \\
Recently, two graph-based pilot assignment schemes have been used: the \emph{graph-colouring} based pilot assignment scheme \cite{liu2020graph} and the \emph{Max-$k$-Cut} based pilot assignment scheme in a weighted graphic framework \cite{ZHLW2021}. In both cases, the CF-mMIMO architecture is modelled as a graph, with the UEs forming the vertices and the edges representing the interference between the UEs. The pilot assignment optimization is then solved on this graph. Since both the graph-colouring and the Max-$k$-Cut problems fall in the class of NP-hard problems, heuristic graph algorithms have been employed to find reasonable solutions to \textsc{PA}, supported by experimental results obtained via simulations, where they outperform other non-graph theoretic algorithms with reasonable time complexity. Further, the Max-$k$-Cut scheme even outperforms the graph-colouring scheme \cite{liu2020graph,ZHLW2021}, making it the state of the art. This led us to believe that there must be some sort of natural relation between the pilot assignment optimization problem and classical graph-theoretic optimization problems. The existing experimental results in the literature aim to optimize a certain Quality of Service (QoS) metric such as the sum-user SE, the system throughput, and the system Energy Efficiency (EE), all of which can be expressed via the uplink/downlink achievable rates, which further depend on the Signal-to-Interference-plus-Noise Ratio (SINR) \cite{BDFZF2021,CHEN2022695,LGX2021,liu2020graph,ZHLW2021}.

\paragraph{Our contribution.}
Despite all the research effort that has been put into solving the pilot assignment problem in cell-free massive MIMO systems, \textsc{PA} has received very little attention from a theoretical computer science perspective. In particular, to the best of our knowledge, there are no complexity results for this problem in the literature. Consequently, in this paper, we show that \textsc{PA} is \emph{strongly} NP-hard via a reduction from \textsc{Min-$k$-Partition}. We further show that the problem does not admit a polynomial-time constant-factor approximation algorithm. Inspired by the results in \cite{Eisenblaetter2001}, we show that \textsc{PA} cannot be approximated in polynomial time within $\mathcal{O}(K^2)$ (where $K$ is the number of users)  when there are at least three pilots in the system. We also present an approximation lower bound of $1.058$ (resp.~$\epsilon|K|^2$, for $\epsilon >0$) in special cases where the system consists of exactly two (resp.~three) pilots. The implication of our result is that any positive (explicit bounds on performance ratios) or negative approximation results for \textsc{Min-$k$-Partition} can be directly translated for \textsc{PA}.

\paragraph{Organization.} The remainder of this paper is organized as follows: We present preliminary definitions, and formally define the CF-mMIMO system as well as the \textsc{PA} optimization problem in Section \ref{sec2}. We present our complexity results on \textsc{PA} in Section \ref{sec3}. Finally, in Section \ref{sec4}, we present some concluding remarks and potential directions for future work.

\section{Preliminaries and problem definition} \label{sec2}
\subsection{Definitions}
The results in this paper use well-established terminologies from complexity and approximation theory. While we recall the most important ones for completness, we refer interested readers to \cite[Chapter 1]{Ausiello1999} for the formal definitions of terms such as optimization problems, NP-hardness, strong NP-hardness, approximation classes APX and $F$-APX, and the complete definition of AP-reducibility.

\begin{definition}[\textbf{Optimization problem}]\label{def1}
An \textbf{optimization problem} $\Pi$ is a tuple $(I_{\Pi}, \, SOL_{\Pi}, \, m_{\Pi}, \, \textup{goal}_{\Pi})$ where:
\begin{itemize}
    \item $I_{\Pi}$ is the set of instances of $\Pi$,
    \item $SOL_{\Pi}$ is a function that associates to an instance $x \in I_{\Pi}$, the set of feasible solutions of $x$,
    \item $m_{\Pi}$ is a measure function, defined for a pair $(x,y)$ where $x \in I_{\Pi}$ and $y \in SOL_{\Pi}(x)$, that returns a positive rational which is the value of the feasible solution $y$.
    \item $\textup{goal}_{\Pi} \in \{\textsc{min, max}\}$ denotes whether $\Pi$ is a minimization or maximization problem.
\end{itemize}    
\end{definition}
\noindent
Given an input instance $x$, we will denote by $SOL_P^*(x)$ the set of optimal solutions of $x$. The value of any optimal solution $y^*(x)$ of $x$ will be denoted by $m^*_{\Pi}(x)$.

\begin{definition}[\textbf{The class NPO}]\label{def2}
 An optimization problem $\Pi = (I_{\Pi}, \, SOL_{\Pi}, \, \\ m_{\Pi}, \, \textup{goal}_{\Pi})$, where the tuple consists of problem instances, feasible solutions for every instance, measure function associated with every feasible solution, and goal of optimization respectively, belongs to the class \textup{\textbf{NPO}} if the following holds: 
 \begin{enumerate}
     \item the set of instances $I_{\Pi}$ is recognizable in polynomial time,
     \item there exists a polynomial $q$ such that, given an instance $x \in I_{\Pi}$, 
     \begin{enumerate}
         \item $|y| \leq q(|x|) \ \ \forall \; y \in SOL_{\Pi}(x)$,
         \item and it is decidable in polynomial time whether $y \in SOL_{\Pi}(x)$ for every $y$ with $|y| \leq q(|x|)$.
     \end{enumerate} 
     \item the measure function $m_{\Pi}$ is computable in polynomial time. 
 \end{enumerate}   
\end{definition}

\begin{definition}[\textbf{NP-hard optimization problems}] \label{def3}
An optimization problem $\Pi$ is called \textup{\textbf{NP-hard}} if, for every decision problem $\Pi' \in \textup{NP}$, $\Pi' \leq_T^p \Pi$, that is, $\Pi'$ can be solved in polynomial time by an algorithm which queries an ``oracle'' that, for any instance $x \in I_{\Pi}$, returns an optimal solution $y^*(x)$ of $x$ along with its value $m_{\Pi}^*(x)$.
\end{definition}

\begin{definition}[\textbf{Strong NP-hardness}]\label{def4}
Let $\Pi$ be a decision problem in NP, with an instance $x \in I_{\Pi}$. Let $\textup{max}(x)$ denote the value of the largest number occurring in the instance $x$. Given a polynomial $p$, let $\Pi^{\textup{max},p}$ be the restriction of $\Pi$ to those instances with the property that $\textup{max}(x) \leq p(|x|)$. If $\Pi^{\textup{max},p}$ remains an \textup{NP-hard} problem for some polynomial $p$, then $\Pi$ is called \textbf{strongly \textup{NP-hard}}.     
\end{definition}

\begin{definition}[\textbf{Performance ratio}] \label{def5}
 Given an optimization problem $\Pi$, for any instance $x$ of $\Pi$ and for any feasible solution $y$ of $x$, the \textbf{performance ratio} of $y$ with respect to $x$ is defined as \[ R(x,y) = \max \left(\frac{m(x,y)}{m^*(x)}, \frac{m^*(x)}{m(x,y)}\right).\] 
\end{definition}

\begin{definition}[\textbf{$\bm{r(n)}$-approximate algorithm}] \label{def6}
Given an optimization problem $\Pi$ in \textup{NPO}, an approximation algorithm $\mathcal{A}$ for $\Pi$, and a function $r : \mathbb{N} \mapsto (1, \infty)$, we say that $\mathcal{A}$ is an \textbf{$\bm{r(n)}$-approximate algorithm} for $\Pi$ if, for any instance $x$ such that $SOL(x) \neq \emptyset$, the performance ratio of the feasible solution $\mathcal{A}(x)$ with respect to $x$ satisfies the following inequality: \[R(x, \mathcal{A}(x)) \leq r(|x|).\]
\end{definition}

\begin{definition}[\textbf{Class APX}] \label{def7}
The class of all \textup{NPO} problems $\Pi$ such that, for some fixed constant $c > 1$, there exists a polynomial-time $c$-approximate algorithm $($also called constant-factor approximation algorithm$)$ for $\Pi$. \\[-2ex]
\end{definition}

\begin{definition}[\textbf{Class \textup{$F$-APX}}] \label{def8} 
Given a class of functions $F$, \textup{$F$-APX} is the class of all \text{NPO} problems $\Pi$ such that, for some function $r \in F$, there exists a polynomial-time $r(n)$-approximate algorithm for $\Pi$.
\end{definition}

\noindent
In the context of approximation algorithms, a reduction from a problem $\Pi_1$ to a problem $\Pi_2$ should guarantee that an approximate solution for $\Pi_2$ yields an approximate solution for $\Pi_1$. Thus we need not only a function $f$ mapping instances of $\Pi_1$ into instances of $\Pi_2$, but also a function $g$ mapping back solutions of $\Pi_2$ into solutions of $\Pi_1$. Next, we define an approximation-preserving reducibility called \emph{AP-reducibility}.

\begin{definition}[\textbf{AP-reducibility}]\label{def9}
  Let $P_1$ and $P_2$ be two optimization problems in NPO. $P_1$ is said to be \textbf{AP-reducible} to $P_2$, written $P_1 \leq_{AP} P_2$, if there exist two functions $f$ and $g$ and a constant $\alpha \geq 1$ such that:
 \begin{enumerate} [itemsep=0pt, topsep=0pt]
     \item For an instance $x \in I_{P_1}$, and for any rational $r > 1$, $f(x,r) \in I_{P_2}$.
     \item For an instance $x \in I_{P_1}$, and for any rational $r > 1$, if $SOL_{P_1}(x) \neq \emptyset$, then \linebreak $SOL_{P_2}(f(x,r)) \neq \emptyset$. 
     \item For any instance $x \in I_{P_1}$, for any rational $r > 1$, and for any $y \in SOL_{P_2}(f(x,r))$, $g(x,y,r) \in SOL_{P_1}(x)$.
     \item $f$ and $g$ are computable by two algorithms $\mathcal{A}_f$ and $\mathcal{A}_g$, respectively, whose running time is polynomial for any fixed rational $r$. 
     \item For any instance $x \in I_{P_1}$, for any rational $r > 1$, and for any $y \in SOL_{P_2}(f(x,r))$, \[R_{P_2}(f(x,r),y) \leq r \: \Rightarrow \: R_{P_1}(x, g(x,y,r)) \leq 1 + \alpha(r-1).\]
     In the rest of this paper, this condition will be referred to as the \textbf{AP-condition}.
 \end{enumerate}
 The triple $(f, g, \alpha)$ is said to be an \textbf{AP-reduction} from $P_1$ to $P_2$. 
\end{definition}

\begin{remark}\label{rem1}
 In most reductions from one optimization problem to another in the literature, the quality of the solution we are looking for is not required to be known explicitly. This is the case with our reductions as well, so we shall replace $f(x,r)$ and $g(x,y,r)$ with $f(x)$ and $g(x,y)$, respectively.
\end{remark}

\subsection{System Model} \label{subsec22}
Building on the description of \textsc{PA} in cell-free massive MIMO (CF-mMIMO) in Section \ref{subsec12}, we remark that the channel estimation error caused by pilot contamination translates into affected achievable rates and eventually leads to an observable degradation of system throughput. Thus, the system performance of cell-free massive MIMO has been characterized by the system throughput in the literature \cite{ZHLW2021}, which is defined as $ \sum\limits_{k=1}^K R_k^{u}$,
where $K$ is the number of users and $R_k^{u}$ denotes the uplink achievable rate for user $k$. The quantity $R_k^{u}$ is further defined as $R_k^{u} =$
\begin{equation}\label{eq1}
     \frac{1-\tau/\tau_c}{2}\log_{2}\left(1 + \frac{\rho^u\eta_k \Bigl( \sum_{m \in A(k)} \gamma_{km} \Bigl)^2}{\splitfrac{\rho^u\sum_{k' \in O(k)}\eta_{k'} \Bigl(\sum_{m \in A(k)} \gamma_{km}\frac{\beta_{k'm}}{\beta_{km}}\Bigl)^2} {+ \rho^u\sum_{k'=1}^K \eta_{k'}\sum_{m \in A(k)} \gamma_{km}\beta_{k'm} + \sum_{m \in A(k)}\gamma_{km}}} \right)
\end{equation}
where 
\renewcommand{\labelitemi}{\textperiodcentered}
\begin{itemize}[itemsep = 0pt, topsep=-3pt]
    \item $\eta_{k}$ is the uplink power control coefficient,
    \item $\rho^u$ is the normalized uplink SNR (signal-to-noise ratio),
    \item $\beta_{km}$ denotes the large-scale fading coefficient between user $k$ and AP $m$ including geometric path loss and shadowing, 
    \item $\gamma_{km}$ denotes the mean-square of the channel estimation of the channel coefficient between user $k$ and AP $m$, 
    \item $A(k)$ denotes the indices of the APs serving user $k$, and
    \item $O(k)$ denotes the set of indices of users $k'$ with the same pilot as user $k$, excluding user $k$. 
\end{itemize}  

\subsection{Problem formulation} \label{subsec23}
We present the first mathematical definition of a CF-mMIMO system as well as a feasible pilot assignment for the system.

\begin{definition}[\textbf{Cell-Free Massive MIMO system}] \label{def10}
Let $\mathcal{A}$ denote the set of APs with cardinality $M$, $U$ denote the set of users with cardinality $K$, and $\Psi$ denote the set of pilots with cardinality $\tau$. Let $\bm{\beta}$ be a $K \times M$ matrix, with the $(k,m)$-th element denoting the large-scale coefficient $\beta_{km}$ between user $k$ and AP $m$. For each user $k$, the set $A(k) \; (1 \leq k \leq K)$ denotes the indices of the subset of APs serving it, as described above. We shall refer to the tuple $\left(\mathcal{A},U,\{A(k)\}_{k = 1}^{K},\bm{\beta},\Psi\right)$ as a \textbf{cell-free massive MIMO (CF-mMIMO) system S}.   
\end{definition}

\begin{definition}[\textbf{Feasible Pilot Assignment}] \label{def11}
A \textbf{feasible} pilot assignment for a CF-mMIMO system $S$ is a well-defined, surjective function $f : U \rightarrow \Psi$. 
\end{definition}

\noindent We state an easy-to-see lemma:

\begin{lemma}\label{lem1}
Finding feasible pilot assignments for the users in a given cell-free massive MIMO system $S$ can be done in polynomial time. 
\end{lemma}
\begin{proof}  
All we really need is the partition of the set of users $U$ such that no two subsets in the partition intersect. As long as $K \gg \tau$, this is always possible. A simple greedy approach to see that this is true would be to randomly select $\tau - 1$ users from $U$, and assign them to the first $\tau - 1$ pilots in $\Psi$ respectively. Then assign the remaining $K - \tau + 1$ users to the remaining $\tau$-th pilot in $\Psi$. By construction, this is a feasible pilot assignment for the system $S$.  \qed
\end{proof}

We now motivate our definition of the Pilot Assignment problem:\\

As mentioned earlier in Introduction Section \ref{subsec12}, we assume that AP selection for each user has already been done. Thus, every quantity in Equation (\ref{eq1}) is a constant. The pilot assignment problem seeks to find the optimum sets $O(k)$ for all $k$. Put simply, we seek a partition of the users into $\tau$ non-empty disjoint sets so that the pilot contamination in the system is minimized. How pilot contamination can be defined mathematically is still a topic of great interest. In this vein, notice that only the first term in the sum in the denominator of the huge fraction inside the logarithm expression in Equation (\ref{eq1}) changes with a change in the assignment of pilots to users. Further, as noted before, only a single term in the denominator changes, as the rest are all constants. Thus we focus on the term \[ \rho^u\sum_{k' \in O(k)}\eta_{k'} \Bigl(\sum_{m \in A(k)} \gamma_{km}\frac{\beta_{k'm}}{\beta_{km}}\Bigl)^2 \text{.} \] As the logarithm is an increasing function, a decrease in this term causes the entire expression to increase. In many well-cited papers in the field \cite{ngo2017cell,liu2020graph,ZHLW2021,ZCZQY2018}, the pilot contamination effect at the $k$-th user is denoted by the term \[\sum\limits_{k' \in O(k)} \sum\limits_{m \in A(k)} \left(\beta_{k'm}/\beta_{km}\right)^2 \text{.}\] 
We understand here that dropping all the constants multiplied with this term technically reduces it to a \emph{simpler} form of the pilot contamination problem. Thus, in this paper, we prove hardness results for this definition of the pilot contamination problem. Since our goal is to minimize pilot contamination in the system with an optimum pilot assignment for all the users, we shall consider the sum of the above terms for all the users, i.e., we aim to minimize \[\sum\limits_{k=1}^{K}\sum\limits_{k' \in O(k)} \sum\limits_{m \in A(k)} \left(\beta_{k'm}/\beta_{km}\right)^2 \text{.} \]

There is also some physical intuition behind dropping the above constants. The large-scale coefficient for a user and an AP is higher when they are geographically closer \cite{schwengler2019wireless}. It is safe to assume that for a user $k$, AP selection yields only those APs $m$ in the set $A(k)$ for which the large-scale coefficients $\beta_{km}$ are reasonably high. The above equations tell us that given a user $k$, the potential disturbance from users $k'$ is severe when the proximity between the interfering users $k'$ and the APs in the set $A(k)$ is higher than that between the user $k$ and the APs in $A(k)$. In a way, we want to put those users $k'$ in $O(k)$, whose large-scale coefficients with respect to the APs that have been selected to serve $k$ are low. In other words, they are farther from those APs than user $k$. \\

To summarize the development so far, we have that given $M$ APs, $K$ users, with each user assigned a set $A(k) \; (1 \leq k \leq K)$ which denotes the indices of the subset of APs serving it, large-scale coefficients $\beta_{km}$ between a user $k$ and AP $m$, and $\tau$ pilots, we need to find a partition of the $K$ users into $\tau$ disjoint sets (where a set contains the users served by the same pilot) such that 
\begin{equation}\label{eq4}
    \sum\limits_{k=1}^{K}\sum\limits_{k' \in O(k)} \sum\limits_{m \in A(k)} \left(\beta_{k'm}/\beta_{km}\right)^2
\end{equation}
is minimized, where $O(k)$ is the set of indices users $k'$ with the same pilot as user $k$, excluding user $k$. Note that since AP selection is done before pilot assignment, the innermost summation remains untouched in this optimization problem. \\

Now let the $\tau$ subsets of $U$ induced by $f$ be $V_1, V_2, \dots V_{\tau}$. In Equation (\ref{eq4}), observe the outer two summations. Given a user $k$, we only have to consider those users $k'$ which are served by the same pilot as it. In other words, given a set in a partition, we must look at all possible pairs of users within it. We must then look at the contribution of the \emph{pair} $k, k'$ to the second sum, which by symmetry turns out to be 
\begin{equation} \label{eq5}
    \sum\limits_{m \in A(k)} \left(\beta_{k'm}/\beta_{km}\right)^2 + \sum\limits_{m' \in A(k')} \left(\beta_{km'}/\beta_{k'm'}\right)^2
\end{equation}
As mentioned in \cite{ZHLW2021}, we can regard the above term as the quantity of interference (thus leading to potential pilot contamination) between users $k$ and $k'$. Since the order or permutation of the elements of the pair $k, k'$ matters, each pair of elements in a partition contributes \emph{two} terms to the second sum. \\

We then rewrite the outer two summations in Equation (\ref{eq4}) as a summation over all possible pairs $k, k'$ of Equation (\ref{eq5}), and then a summation of this over all sets in the partition:

\begin{equation} \label{eq6}
    \sum\limits_{t=1}^{\tau}\sum\limits_{k, k' \in V_t} \Biggl(\sum\limits_{m \in A(k)} \left(\beta_{k'm}/\beta_{km}\right)^2 + \sum\limits_{m' \in A(k')} \left(\beta_{km'}/\beta_{k'm'}\right)^2 \Biggl) \\[1.5ex]
\end{equation}

Thus, we get the following definition for the Pilot Assignment optimization problem: 

\begin{definition}[\textbf{Pilot Assignment Problem (\textsc{PA})}]\label{def12}
Given a cell-free massive MIMO system $S$, we call the optimization problem \[ \min\limits_{f feasible}\sum\limits_{\substack{k,k' \in U \\ f(k) = f(k')}} \Biggl(\sum\limits_{m \in A(k)} \left(\beta_{k'm}/\beta_{km}\right)^2 + \sum\limits_{m' \in A(k')} \left(\beta_{km'}/\beta_{k'm'}\right)^2 \Biggl) \] the \textbf{Pilot Assignment \textup{(\textsc{PA})}} problem. \\[-1ex]
\end{definition}
\noindent We now define a classical graph problem that shall be crucial in proving our results.

\begin{definition}[\textbf{\textsc{Min-$k$-Partition}}]\label{def13}
 Given an undirected graph $G = (V, E)$ with $n$ vertices and weight $\omega_{i,j} \in \mathbb{Q}_+$ for the edge joining vertices $i$ and $j \ \ \forall \; 1 \leq i,j \leq n$, the \textsc{Min-$k$-Partition} problem seeks to find a partition $\mathcal{V}$ of $V$ into $k$ disjoint sets $\{V_1, V_2, \dots, V_{k}\}$ such that the total weight of the edges with endpoints within the same set is minimum.    
\end{definition}

\noindent The objective of the above problem can be formulated as: 
\begin{equation} \label{eq2}
    \min\limits_{\mathcal{V} = \{V_1, V_2, \dots, V_{k}\}} \hspace{1mm} \sum\limits_{p=1}^{k}\hspace{1.5mm} \sum\limits_{i, j \in V_p} \omega_{i,j}
\end{equation}
\noindent Note that the \textsc{Max-$k$-Cut} problem referred to in Section \ref{subsec13} is the dual of the \textsc{Min-$k$-Partition} problem. Moreover, both the \textsc{Min-$k$-Partition} and the \textsc{Max-$k$-Cut} problems are known to be NP-hard \cite{ALMSS1998,Chopra1993ThePP,GJ1990}.  

 We shall formally define both the optimization and decision versions of the problem as per Definition \ref{def1} in Appendix \ref{apa}. These can be referred to if the reader wants an explicit definition of \textsc{PA}  as a tuple. 

\section{Complexity results for {\textsc PA}} \label{sec3}

\subsection{Primary results on complexity} \label{subsec31}
To talk about the computational complexity of {\textsc PA}, we assume that all numerical values appearing as input data are rational and that they are encoded in binary form. We now present our first major theorem.

\begin{theorem}\label{thm1}
The following results hold for the time complexity of \textsc{PA}: 
\begin{enumerate}[label={$(\roman*)$}, ref=$\thetheorem.($\roman*$)$]
\item \textsc{PA} $\in$ \textup{NPO}. \label{thm11}
\item \textsc{PA} is \textup{NP-hard}. \label{thm12}
\item \textsc{PA} is strongly \textup{NP-hard}. \label{thm13}
\end{enumerate}
\end{theorem}

\begin{proof}
   \begin{enumerate}[label={$(\roman*)$}]
    \item[]

    \item It is recognizable in polynomial time whether a string encodes a tuple representing a cell-free massive MIMO system, as it primarily involves a check of the cardinality of multiple sets, and the dimension of a matrix. If all the users and pilots are listed individually, then the encoding length of a feasible pilot assignment for the users cannot exceed the encoding length of the CF-mMIMO system. From Lemma \ref{lem1}, we know that it is decidable in polynomial time if a pilot assignment is feasible. Finally, it is clear that Equation (\ref{eq6}) is computable in polynomial time. Therefore, \textsc{PA} $\in$ NPO. \\
    
    \item We now present our proof of part $(ii)$, i.e., \textsc {PA} is NP-hard, by establishing the following claim: \textsc{Min-$k$-Partition} $\leq_T^p$ \textsc{PA}.
  \\ \\
    Consider an arbitrary instance of the \textsc{Min-$k$-Partition} problem, specified by a weight\-ed graph $G = (V, E, \omega)$, where $\omega : E \rightarrow \mathbb{R}_+$ is a function mapping edges to their weights. We construct an instance of \textsc{PA} as follows: \\ \\    
    Set $K = |V|$. The set of users $U$ is set to be the set of vertices $V$. Thus, $|U| = K$. The number of pilots $\tau$ is set to be $k$. This determines the set $\Psi = \{1, \dots, k\}$. For all users indexed by $1 \leq i \leq K$, set $|A(i)| = 1$. This determines $K$ number of APs. Since we know that in a practical cell-free massive MIMO system, we have $M \gg K$, we could define many dummy APs to achieve this condition. We set the total number of APs in our \textsc{PA} instance as some arbitrarily large constant $M \gg K$, which determines the set $\mathcal{A} = \{1, \dots, K, \dots, M\}$. Further, $\forall \; 1 \leq i \leq K$, let $A(i) = \{i\}$. Thus, a single, distinct AP indexed by $i$ serves the user $i$. Now, $\forall \; 1 \leq i \leq K$ and  $\forall \; 1 \leq m \leq M$ set
    
    $$ \beta_{im} =
        \begin{cases}
        1     & \text{if $m = i$,}\\
        \sqrt{\frac{\omega((i, m))}{2}}     & \text{if $m \neq i$ and $m \leq K$,}\\
        0  & \text{otherwise.}
        \end{cases}$$
    This determines the matrix $\bm{\beta}$ in our instance of the problem. 
    \\ \\
    $(\Rightarrow)$ If we have a solution to an instance of the \textsc{Min-$k$-Partition} problem, then by Definition \ref{def13}, we have a partition $\mathcal{V} = \{V_1, V_2, \dots, V_{k}\}$ of the vertex set $V$ such that the expression $\sum\limits_{t=1}^{k}\sum\limits_{i,j \in V_t} \omega((i, j))$ is minimized. By our construction, the set of vertices $V$ is the set of users $U$ and $k$ is the cardinality of the set $\Psi$, $\tau$. Replacing $k$ by $\tau$ and $i, j$ by the arbitrary indices $k, k'$ to denote the users in our constructed instance of the PA problem, we get that the minimized equation is $\sum\limits_{t=1}^{\tau}\sum\limits_{k,k' \in V_t} \omega((k, k'))$. Recall that the objective function to be minimized in a general instance of \textsc{PA} is Equation (\ref{eq6}), which can be simplified to 
    \begin{align*} 
        &\sum\limits_{t=1}^{\tau}\sum\limits_{k, k' \in V_t} \Biggl(\sum\limits_{m \in A(k)} \left(\beta_{k'm}/\beta_{km}\right)^2 + \sum\limits_{m' \in A(k')} \left(\beta_{km'}/\beta_{k'm'}\right)^2 \Biggl) \\ 
        &= \sum\limits_{t=1}^{\tau}\sum\limits_{k, k' \in V_t} \Biggl( \left(\beta_{k'k}/\beta_{kk}\right)^2 + \left(\beta_{kk'}/\beta_{k'k'}\right)^2 \Biggl)
        \\
        &= \sum\limits_{t=1}^{\tau}\sum\limits_{k, k' \in V_t} \omega((k, k')) 
    \end{align*}
    \noindent   Thus we have a solution to our \textsc{PA} problem instance. \\ \\    
    $(\Leftarrow)$ By the above simplification of Equation (\ref{eq6}), if we have a solution for the above-constructed instance of \textsc{PA}, we end up minimizing the expression \[\sum\limits_{t=1}^{\tau}\sum\limits_{k, k' \in V_t} \omega((k, k'))\] Now, we map the users $k, k'$ back to the vertices of the graph via arbitrary indices $i, j$, and note that the number of pilots $\tau$ is in fact the desired number of cuts, $k$. So we end up minimizing the expression \[\sum\limits_{t=1}^{k}\sum\limits_{i,j \in V_t} \omega((i, j))\] which represents the total weight of such edges which have both endpoints in the same partition. Thus, due to the construction of our \textsc{PA} instance, we also have a solution to the corresponding \textsc{Min-$k$-Partition} instance.  
    \\ \\
    Therefore, the \textsc{Min-$k$-Partition} instance has a solution \emph{if and only if} if the above \textsc{PA} instance has a solution. This proves our claim. Since \textsc{Min-$k$-Partition} is NP-hard, we conclude that \textsc{PA} is NP-hard. \\

    \item It is stated in \cite{FAIRBROTHER201797} that \textsc{Min-$k$-Partition} is strongly NP-hard due to a reduction from the \textsc{Graph $k$-Colourability} problem, which has been proven to be strongly NP-complete \cite{GJ1978}. As we could not find a reference for an explicit proof of this fact in the literature, we give a simple proof of the aforementioned reduction:
    The \textsc{Chromatic Number} or \textsc{Graph $k$-Colourability} is a non-numeric decision problem which asks whether the vertices of a graph $G$ can be coloured using at most $k$ colours such that no two adjacent vertices have the same colour. In other words, it asks if we can partition the vertices into at most $k$ sets such that each of the sets is an independent set. \\
    
    We give a simple Turing reduction from  \textsc{Graph $k$-Colourability} to \textsc{Min-$k$-Partit\-ion} as follows: Given an instance $(G, E, k)$ of the \textsc{Graph $k$-Colou\-rability} problem, associate to it an instance of \textsc{Min-$k$-Partition} defined by $G = (V, E, \omega)$, such that $\omega : E \rightarrow 1$. It is easy to see that $G$ is $k$-colourable if and only if the optimal solution to the associated instance of \textsc{Min-$k$-Partition} yields a value of $0$. This is due to the fact that a solution to the \textsc{$k$-Colouring} problem on $G$ induces $k$ independent sets in $G$. In our instance of \textsc{Min-$k$-Partition}, these $k$ independent sets translate to $k$ subsets of the set of vertices $V$ of the graph, such that no two vertices in a given subset are adjacent. Thus, the endpoints of any edge in $G$ must be in two different sets. This gives us the minimum possible value of the sum of the weights of the edges with endpoints in the same partition: $0$. Thus, we see that independent sets in $G$ give us the optimum solution to the constructed instance of \textsc{Min-$k$-Partition}. On the other hand, if the solution to our \textsc{Min-$k$-Partition} instance has a value of $0$, then all the edges in $G$ have endpoints in disjoint sets of the partition. This follows from the fact that the weight of each edge is 1. Thus, the vertices in a given subset are pairwise disjoint, forming an independent set. Colouring the $k$ independent sets in this partition of $V$ using $k$ different colours gives us the desired solution to the \textsc{Graph $k$-Colourability} problem.  \qed
    
\end{enumerate}
\end{proof}

\noindent A corollary of the above theorem follows immediately: 
\begin{corollary} \label{cor1}
For every $q \in \mathbb{Q}_+$, the decision version of \textsc{PA} is \textup{NP-complete}. 
\end{corollary}

\subsection{Further results on approximability} \label{subsec32}
The first theorem tells us that unless P $=$ NP, there exists no polynomial-time algorithm to solve \textsc{PA}. We turn our attention to the more practical aspects of tackling this problem. Instead of trying to obtain optimal solutions, we look at how well we can approximate the optimal solution to the problem in polynomial time. Since we are dealing with optimization problems (as opposed to decision problems), we must be careful when giving a polynomial-time reduction from one optimization problem to the other. While the decision problems corresponding to most NP-hard optimization problems admit a standard polynomial-time many-one (or Karp) reduction to each other, the optimization problems do not share the same approximability properties. This is due to the fact that many-one reductions do not always preserve the measure function and, even if they do, they seldom preserve the quality of the solutions. 
This is reflected in the approximability of the equivalent \textsc{Min-$k$-Partition} and \textsc{Max-$k$-Cut} problems as well. We know that while \textsc{Max-$k$-Cut} is APX-complete,  \textsc{Min-$k$-Partition} is not in APX (see Appendix B of \cite{Ausiello1999}).  \\ \\
Thus, we appeal to Definition \ref{def9} to introduce a stronger kind of reducibility, namely AP-reducibility. Although different types of approximation-preserving reducibilities exist in the literature, AP-reducibility is sufficiently general to incorporate the properties of almost all such reducibilities, while also establishing a \emph{linear relation} between performance ratios.
Approximation-preserving reductions induce an order on optimization problems based on their ``difficulty'' of being approximated. Approximation-preserving reductions are also an essential tool for proving non-approximability results. \\ \\
Next, we state four lemmas, which will help us prove our second theorem.

\begin{lemma}[\cite{Ausiello1999}]\label{lem2}
If \textup{P $\neq$ NP}, \textup{exp-APX $\subset$ NPO}.  
\end{lemma}
\begin{proof}
While this fact is stated in \cite{Ausiello1999}, we give a detailed proof (with an example mentioned in \cite{Ausiello1999}) as follows: \\

 Let $P = (I, \, SOL, \, m, \, \textup{goal})$ be a problem in NPO. Since $m$  is computable in polynomial time, there exist $h$ and $k$ such that for any $x \in I$ with $|x| = n$ and for any $y \in SOL(x)$, $m(x,y) \leq h2^{n^k}$. This is because the range of possible values of $m(x,y)$ has an upper bound given by $M = 2^{p(|x|)}$ for some polynomial $p$, which is again due to the properties of NPO problems which state that the length $|y|$ of any solution $y \in SOL(x)$ is bounded by $q(|x|)$ for some polynomial $q$, and $m$ is computable in polynomial time (see Definition \ref{def9}). This implies that any feasible solution has a performance ratio bounded by $h2^{n^k}$. Indeed, the polynomial bound on the computation time of the measure function for all NPO problems implies that they are $h2^{n^k}$-approximable for some $h$ and $k$. This seems to imply that the classes exp-APX and NPO are the same. However, we note that there exist several problems in NPO for which it is hard even to decide whether any feasible solution exists, (and thus to find such a feasible solution) unless P $=$ NP. An example of such a problem is the \textsc{Minimum $\{0,1\}$-Linear Programming} problem, which belongs to NPO. The problem instance consists of a matrix $A \in \mathbb{Z}^{m \times n}$ and vectors $b \in \mathbb{Z}^{m}$, $w \in \mathbb{N}^{n}$. The problem asks for a solution $x \in \{0,1\}^n$ such that $Ax \geq b$, and the measure function $\sum\limits_{i=1}^n w_ix_i$ is minimized. Given an integer matrix $A$ and an integer vector $b$, deciding whether a binary vector $x$ exists such that $Ax \geq b$ is NP-hard, as the \textsc{Satisfiability} problem is polynomial-time reducible to this decision problem (see Example 1.10 in Section 1.3.1, Chapter 1 of \cite{Ausiello1999}). This implies that, if P $\neq$ NP, then \textsc{Minimum $\{0,1\}$-Linear Programming} does not belong to exp-APX.  \qed 
\end{proof}
\vspace{1.5mm}

\begin{lemma} \label{lem3}
The polynomial-time reduction from \textsc{Min-$k$-Partition} to \textsc{PA} given in Theorem \ref{thm12} is an AP-reduction with $\alpha = 1$. 
\end{lemma}
\begin{proof}
Consider an instance $G = (V, E, \omega)$ of \textsc{Min-$k$-Partition}, we need to determine the function $f$ that maps it to an instance of \textsc{PA}. A general instance of \textsc{PA} is determined by the tuple $\left(\mathcal{A},U,\{A(k)\}_{k = 1}^{K},\bm{\beta},\Psi\right)$. Recall the definition of AP-reduction from Definition \ref{def9}. From our reduction, we get that $f : (V, E, \omega) \mapsto \left(\{1, \dots, |V|, \dots, M\}, \, V, \, \{i\}_{i = 1}^{|V|}, \, \bm{\beta}, \, \{1, \dots, k\}\right)$, where $M$ is an arbitrary constant such that $M \gg |V|$, and the $im$-th element of the matrix $\bm{\beta}$ is $\beta_{im} = 1$ if $m = i$, $\sqrt{\frac{\omega((i, m))}{2}}$ if $m \neq i$ and $m \leq K$, and 0 otherwise. Moreover, if $h$ is the feasible assignment that forms the solution of our constructed instance of \textsc{PA}, then the partition of vertices $\mathcal{V} = \{V_1, V_2, \dots, V_k\}$ that forms the solution of our original \textsc{Min-$k$-Partition} instance is defined as $V_t = \{i \, | \, h(i) = t \}$ where $i \in V$. We set $g : (G, h) \mapsto \mathcal{V}$, where  $V_t = \{k \, | \, k \in U \; \text{and} \; h(k) = t \} \quad \forall \; 1 \leq t \leq \tau$.  Notice that $f$ and $g$ do not depend on the performance ratio. \\

Finally, for any instance $G$ of \textsc{Min-$k$-Partition}, with $S = f(G)$ as the constructed instance of \textsc{PA}, $h$ as the solution of $S$, and $\mathcal{V}= g(G, \mathcal{V})$ as the solution to the original instance $G$ of \textsc{Min-$k$-Partition}, we see that $m_{\textsc{M$k$P}}(G, \mathcal{V}) = m_{\textsc{PA}}(S, h)$ (where \textsc{M$k$P} is short for \textsc{Min-$k$-Partition}), by the nature of the constructed instances and solutions. We conclude that $f$ and $g$ satisfy the AP-condition with $\alpha = 1$, and thus the polynomial-time reduction described above is an AP-reduction.  \qed
\end{proof}

\vspace{1.5mm}

\begin{lemma}\label{lem4}
\textup{\textsc{PA} $\leq_{AP}$ \textsc{Min-$k$-Partition}} with $\alpha =1$. 
\end{lemma}
\begin{proof}
    Since we've proved that the decision version of \textsc{PA} is NP-complete, it follows from the definition of NP-completeness that \textsc{PA} and \textsc{Min-$k$-Partition} are equivalently hard. We first give an explicit polynomial-time reduction from \textsc{PA} to \textsc{Min-$k$-Partition}, and then prove that it is in fact an AP-reduction. \\

Consider an arbitrary instance of \textsc{PA}, specified by the cell-free massive MIMO system $S = (\mathcal{A},U,A(k),\bm{\beta},\Psi)$ with $|\mathcal{A}| = M, \, |U| = K, \, |\Psi| = \tau$. 
We construct an instance of \textsc{Min-$k$-Partition} as follows:
Define a \emph{complete} graph $G = (V, E, \omega)$ (where $\omega : E \rightarrow \mathbb{R}_+$ is a function that maps the edges in our graph to positive real values), by setting $V = U$. By construction, we have $|V| = K$ and $|E| = \frac{K(K+1)}{2}$. For any edge $(k,k') \in E$, where $k, k' \in V$, we set \[\omega((k, k')) =  \sum\limits_{m \in A(k)} \left(\beta_{k',m}/\beta_{k,m}\right)^2 + \sum\limits_{m' \in A(k')} \left(\beta_{k,m'}/\beta_{k',m'}\right)^2. \] The weight function $\omega$ defined above satisfies a symmetric property, i.e., $\omega((k, k')) = \omega((k', k))$. It is important to note that we set the value of $k$ that is referred to in the title of the \textsc{Min-$k$-Partition} problem to $\tau$. So we have constructed an instance of a \textsc{Min-$\tau$-Partition} problem. \\

($\Rightarrow$) If we have a solution to an instance of \textsc{PA}, then we have minimized the optimization function given by Equation (\ref{eq6}), which we state again:  \[ \sum\limits_{t=1}^{\tau}\sum\limits_{k, k' \in V_t} \Biggl(\sum\limits_{m \in A(k)} \left(\beta_{k'm}/\beta_{km}\right)^2 + \sum\limits_{m' \in A(k')} \left(\beta_{km'}/\beta_{k'm'}\right)^2 \Biggl) \]
But notice that the objective function to be minimized in a general instance of the \textsc{Min-$k$-Partition} problem is $\sum\limits_{t=1}^{k}\sum\limits_{i,j \in V_t} \omega((i, j))$. By the construction of our instance of the \textsc{Min-$k$-Partition} problem, the users $U$ form the vertices $V$ of the graph and the cardinality $\tau$ of the set of pilots $\Psi$ is $k$ (which is in the title of the problem, denoting the number of subsets of the vertices). Replacing $\tau$ by $k$ and the user indices $k,k'$ by the arbitrary indices $i,j$ to denote the vertices of our graph, we get that the function minimized by solving the \textsc{PA} instance is   
\begin{align*} 
    \sum\limits_{t=1}^{k}\sum\limits_{i, j \in V_t} \Biggl(\sum\limits_{m \in A(i)} \left(\beta_{jm}/\beta_{im}\right)^2 + \sum\limits_{m' \in A(j)} \left(\beta_{im'}/\beta_{jm'}\right)^2 \Biggl)
    =& \sum\limits_{t=1}^{k}\sum\limits_{i,j \in V_t} \omega((i, j))  
\end{align*}
Therefore, we have a solution to our instance of the \textsc{Min-$\tau$-Partition} problem. \\

($\Leftarrow$) By the above argument, it is clear that if we have a solution to our constructed instance of the \textsc{Min-$\tau$-Partition} problem, we also have a solution to the corresponding \textsc{PA} instance by a reverse change of indices, as seen in Theorem \ref{thm12}. \\

Hence, we have shown an explicit polynomial-time reduction from \textsc{PA} to \textsc{Min-$k$-Partition}. What remains to be verified is that this is an AP-reduction. To see this, consider an instance of \textsc{PA} given by $S = \left(\mathcal{A},U,\{A(k)\}_{k = 1}^{K},\bm{\beta},\Psi\right)$. A general instance of \textsc{Min-$k$-Partition} is given by an undirected graph $G = (V, E, \omega)$. In the above reduction, we have that $G$ is a complete graph with $V = U$ and $\omega((i, j)) =  \sum\limits_{m \in A(i)} \left(\beta_{j,m}/\beta_{i,m}\right)^2 + \sum\limits_{m' \in A(j)} \left(\beta_{i,m'}/\beta_{j,m'}\right)^2 \quad \forall \: i,j \in V$. 
Recall Definition \ref{def9}. Thus we have $ f : \left(\mathcal{A},U,\{A(k)\}_{k = 1}^{K},\bm{\beta},\Psi\right) \mapsto \left(U, E(K_U), w \right)$ where $E(K_U)$ denotes the set of edges of the complete graph on the vertices denoted by $U$, and $\omega((k, k')) = \sum\limits_{m \in A(k)} \left(\beta_{k',m}/\beta_{k,m}\right)^2 + \sum\limits_{m' \in A(k')} \left(\beta_{k,m'}/\beta_{k',m'}\right)^2$. Further, if we have obtained the partition $\mathcal{V} = \{V_1, V_2, \dots, V_{k}\}$ of the vertex set $V$ in our constructed instance of \textsc{Min-$k$-Partition}, we get that the feasible assignment $h$ which is the solution to the corresponding \textsc{PA} instance is defined as $h(k) = t \ \text{such that} \; k \in V_t \quad \forall \: k \in U$. We set $g : (S, \mathcal{V}) \mapsto h$, where $h(i) = t \ \text{such that} \ i \in V_t \in \mathcal{V}$. Notice that $f$ and $g$ do not depend on the performance ratio.\\
Finally, for any instance $S$ of \textsc{PA}, with $G = f(S)$ as the constructed instance of \textsc{Min-$k$-Partition}, $\mathcal{V}$ as the solution of $G$, and $h = g(S, \mathcal{V})$ as the solution to the original instance $S$ of \textsc{PA}, we see that $m_{\textsc{PA}}(S, h) = m_{\textsc{M$k$P}}(G, \mathcal{V})$ (where \textsc{M$k$P} is short for \textsc{Min-$k$-Partition}), by the nature of the constructed instances and solutions. We conclude that $f$ and $g$ satisfy the AP-condition with $\alpha = 1$, and thus the polynomial-time reduction described above is an AP-reduction.  \qed
\end{proof}
\vspace{3mm}

Observe that an AP-reduction from \textsc{PA} to \textsc{Min-$k$-Partition} is important from a practical viewpoint to prove positive approximation results for \textsc{PA}. In particular, it allows us to translate the performance ratios that exist for the approximation of special cases of \textsc{Min-$k$-Partition} (e.g., when $k$ is $2$ or $3$) into performance ratios for the corresponding \textsc{PA} problem \cite{KKLP1997}. On the other hand, an AP-reduction from \textsc{Min-$k$-Partition} to \textsc{PA} is necessary to prove negative approximability results for \textsc{PA}. We recall a few results on the approximability of \textsc{Min-$k$-Partition} stated in \cite{Eisenblaetter2001,KKLP1997} in the form of a  lemma.

\vspace{2mm}

\begin{lemma}\label{lem5}
Assuming \textup{P $\neq$ NP}, the following statements hold with respect to the computational complexity of \textsc{Min-$k$-Partition}:
\begin{enumerate}[label={$(\roman*)$}, ref=\thelemma.(\roman*)]
     \item The problem is not in \textup{APX}. \textup{\cite{SG1976}} \label{lem51}
     \item For $k \geq 3$, it is NP-hard to approximate \textup{\textsc{Min-$k$-Partition}} within $O(|E|)$, even when restricting the instances to graphs with $|E| = \Omega(|V|^{2-\epsilon})$, for a fixed $\epsilon$, $0 < \epsilon < 1$. \textup{\cite{KKLP1997}} \label{lem52}
     \item No polynomial time algorithm can achieve a better performance ratio than 1.058 in the case of $k=2$. \textup{\cite{Has2001}} \label{lem53}
    \item In case of $k = 2$, a polynomial time algorithm with a performance guarantee of $\log |V|$ is known. \textup{\cite{GVY1996}}  \label{lem54}
    \item In case of $k=3$, a polynomial time algorithm with a performance guarantee of $\epsilon|V|^2$ for any $\epsilon > 0$ is known. \textup{\cite{KKLP1997}} \label{lem55}  
\end{enumerate}
\end{lemma}
\noindent Using Lemmas \ref{lem2}, \ref{lem3}, \ref{lem4} and \ref{lem5}, we are now ready to prove the following theorem on the approximability of \textsc{PA}, inspired by \cite[Section 3.2.3]{Eisenblaetter2001}.

\begin{theorem}\label{thm2}
Assuming P $\neq$ NP, The following statements hold true for the Pilot Assignment problem:
\begin{enumerate}[label={$(\roman*)$}, ref=\thetheorem.(\roman*)]
    \item \textsc{PA} is not in \textup{APX}, however it is in \textup{exp-APX}.\label{thm21}
    \item An approximation of \textsc{PA} within $\mathcal{O}(K^2)$ for $\tau \geq 3$ is impossible in polynomial time. \label{thm22}
    \item In the special case of $\tau = 2$, while no polynomial time algorithm can achieve a performance ratio better than 1.058, there exists an algorithm with a performance guarantee of $\log |K|$ in polynomial time. \label{thm23}
    \item In the special case of $\tau = 3$, there exists a polynomial time algorithm with a performance guarantee of $\epsilon|K|^2$ for any $\epsilon > 0$. \label{thm24}
\end{enumerate}
\end{theorem}

\begin{proof}
 \begin{enumerate}[label={$(\roman*)$}]
    \item []
    \item From Lemmas \ref{lem1} , \ref{lem2} and \ref{lem51}, the result follows immediately. 
    \item From Lemma \ref{lem52}, we have that \textsc{Min-$k$-Partition} cannot be approximated within $\mathcal{O}(|E|)$ in polynomial time for $k \geq 3$. Further, in \textsc{PA}, we typically deal with complete (or dense) graphs. So we have that $|E| = \Omega(K^2)$, where $K$ is the number of users, or vertices in the graph. The number of pilots available, or the number of partitions required in the graph is $\tau$.  Hence, using Lemma \ref{lem3}, we get that unless P $=$ NP, it is impossible to approximate \textsc{PA} within $\mathcal{O}(K^2)$ in polynomial time for $\tau \geq 3$. 
    
    \item The first part of this claim follows from Lemma \ref{lem3} and Lemma \ref{lem53}, while the second part follows from Lemma \ref{lem4} and Lemma \ref{lem54}.  
    \item As above, this follows easily from Lemma \ref{lem4} and Lemma \ref{lem55}.  \qed 
\end{enumerate}
\end{proof}

\section{Conclusion} \label{sec4}
In this paper, we studied the inherent hardness of the pilot assignment problem and found provable guarantees on the quality of achievable solutions to the problem. We defined the cell-free massive MIMO system mathematically, which is crucial in any theoretical study of this problem. We further proved that \textsc{PA} is strongly NP-hard, and at the same time, does not belong to the APX class of problems. The big picture here is that, to the best of our knowledge, our paper is the first to study \textsc{PA} from a theoretical computer science perspective, providing complexity results for the problem. 
\\ \\
Recall that while the \textsc{Max-$k$-Cut} problem (see Section \ref{subsec13}) has been used to develop the current state of the art for \textsc{PA}, our contribution here is to focus on its dual problem, the \textsc{Min-$k$-Partition}. Due to the varying nature of the approximability of these optimization problems, it is better to apply a heuristic for the \textsc{Min-$k$-Partition} problem to solve \textsc{PA}. This is due to the fact that the performance ratio of any heuristic/algorithm designed for \textsc{Min-$k$-Partition} can be translated directly for \textsc{PA}. There is extensive research on formulating and solving the \textsc{Min-$k$-Partition} problem using Linear Programming (LP) and Semi-definite Programming (SDP) relaxation approaches \cite{Chopra1993ThePP,Eisenblaetter2001,FAIRBROTHER201797,FJ1997,Ghaddar2011ABA,Wang2017ExploitingSF}. This remains an active area of research. Further research in this area seems to be a step in the right direction towards finding better optimal solutions for \textsc{PA}.
\\ \\
It is also interesting to note that there has been very little research on the parameterized complexity of the \textsc{Min-$k$-Partition} problem. This leads us to believe that the problem may not be FPT (fixed-parameter tractable). A practical question that arises here is if this problem admits an XP algorithm for some suitable parameter. Such an algorithm may yield a pilot assignment scheme with superior performance, depending on the parameter used in the algorithm. 

\subsubsection{Acknowledgements.} The authors thank Prof.~David Manlove for his feedback on drafts of this paper. Further, we gratefully acknowledge the support of Dr.~Yusuf Sambo and Dr.~Abdulrahman Al Ayidh in understanding the problem as well as its existing heuristics, during the early stages of this project. This work was done while S. Prusty was working on her Master's project at the School of Computing Science, University of Glasgow. S. Prusty would like to thank DST-INSPIRE, Government of India for their support via the Scholarship for Higher Education (SHE) programme. 

\appendix

\section{Appendix for Section \ref{subsec23}} \label{apa}

\subsection{Definitions of the Decision and Optimization versions of \textsc{PA}} 
\begin{definition}[\textbf{Decision version of \textsc{PA}}]\label{def14}
The decision version of \textsc{PA} is defined by a triple $(I_{\textsc{PA}}, q, SOL_{\textsc{PA}})$, where $I_{\textsc{PA}}$ is the set of all cell-free massive MIMO systems $S$, $q \in \mathbb{Q}_+$ and $SOL_{\textsc{PA}} : I_{\textsc{PA}} \rightarrow \{0,1\}$ is a function that assigns to each pair $(S, q)$ either the value 0 or 1. The assignment is done based on whether an instance $S$ has a value of at most $q$ for Equation $(\ref{eq6})$. The problem asks if there exists an instance $S \in I_{\textsc{PA}}$ such that given a number $q$, $SOL_{\textsc{PA}}(S, q) = 1$.
\end{definition}

\begin{definition}[\textbf{Optimization version of \textsc{PA}}] \label{def15}
The minimization problem \textsc{PA} is characterized by the tuple $(I_{\textsc{PA}}, SOL_{\textsc{PA}}, m_{\textsc{PA}}, \textsc{min})$, where $I_{\textsc{PA}}$ is the set of all cell-free massive MIMO systems $S$ $($instances of \textsc{PA}$)$, $SOL_{\textsc{PA}}$ is a function that assigns every instance $S \in I_{\textsc{PA}}$ to its set of feasible pilot assignments $f$, and $m_{\textsc{PA}}$ is a measure function that assigns to each pair $(S, f)$, where $S \in I_{\textsc{PA}}$ and $f \in SOL_{\textsc{PA}}(S)$, the value of Equation $(\ref{eq6})$. The problem asks for a given instance $S$, a feasible assignment $f \in SOL_{\textsc{PA}}(S)$ such that $m_{\textsc{PA}}(S,f) = \min \limits_{f' \in SOL_{\textsc{PA}}(S)} m_{\textsc{PA}}(S,f')$. 
\end{definition}

\bibliographystyle{splncs04}
\bibliography{biblio}
%

\end{document}